\documentclass[conference]{IEEEtran}
\IEEEoverridecommandlockouts
\usepackage{cite}
\usepackage{amsthm}
\usepackage{amsmath,amssymb,amsfonts}
\usepackage{algorithmic}
\usepackage{graphicx}
\usepackage{textcomp}
\usepackage{xcolor}
\usepackage{cleveref}

\usepackage{enumitem}
\usepackage{etoolbox}
\usepackage{graphicx}
\usepackage{makecell}
\usepackage{caption}
\usepackage{wrapfig}
\usepackage{setspace}
\usepackage{subcaption}
\usepackage{floatflt}
\usepackage{float}
\usepackage{setspace}

\newtheorem{theorem}{Theorem}
\newtheorem{lemma}{Lemma}

\DeclareMathOperator*{\argmax}{argmax}

\begin{document}

\title{Refined Statistical Bounds for Classification Error Mismatches with Constrained Bayes Error
}

\author{
    \IEEEauthorblockN{Zijian Yang\IEEEauthorrefmark{1}\IEEEauthorrefmark{2}, Vahe Eminyan\IEEEauthorrefmark{1}, Ralf Schlüter\IEEEauthorrefmark{1}\IEEEauthorrefmark{2}, Hermann Ney\IEEEauthorrefmark{1}\IEEEauthorrefmark{2}}
    \IEEEauthorblockA{\IEEEauthorrefmark{1}Machine Learning and Human Language Technology Group, Lehrstuhl Informatik 6,\\
Computer Science Department, RWTH Aachen University, Germany}
    \IEEEauthorblockA{\IEEEauthorrefmark{2}AppTek GmbH, Germany}
}
\maketitle

\begin{abstract}
In statistical classification/multiple hypothesis testing and machine learning, a model distribution estimated from the training data is usually applied to replace the unknown true distribution in the Bayes decision rule, which introduces a mismatch between the Bayes error and the model-based classification error.
In this work, we derive the classification error bound to study the relationship between the Kullback-Leibler divergence and the classification error mismatch. We first reconsider the statistical bounds based on classification error mismatch derived in previous works, employing a different method of derivation. Then, motivated by the observation that the Bayes error is typically low in machine learning tasks like speech recognition and pattern recognition, we derive a refined Kullback-Leibler-divergence-based bound on the error mismatch with the constraint that the Bayes error is lower than a threshold.

\end{abstract}


\section{Introduction \& Related Work}

Statistical classification, also known as multiple hypothesis testing, is widely applied in machine learning areas, e.g. neural machine translation \cite{stahlberg2020neural}, automatic speech recognition \cite{prabhavalkar2023end}, and pattern recognition \cite{jain2000statistical}. In these tasks, the recognition/decoding result is generated by using
a decision rule. In statistical classification, an important performance measure is the classification error, which is minimized by the Bayes decision rule that utilizes the underlying true distribution of the classification task. Information theory provides several bounds on the Bayes error, such as Chernoff bound \cite{chu1971error}, nearest neighbor bound \cite{devijver1974new}, and Lainiotis bound \cite{lainiotis1969class}. However, these general bounds do not cover more specific modeling issues. In practice, the true distribution is unknown, and a probabilistic model is trained to approximate the true distribution and used in the decision rule, which introduces a discrepancy between the true distribution of the data and the probabilistic model \cite{ney2003relationship,schluter2013novel}. This discrepancy is not addressed in the works for error bounds on the Bayes error. In this work, we will make a mathematically strict distinction between true and model distributions used in decision rules. We refer to the difference between the Bayes error and the model-based classification error as the classification error mismatch.

In information theory and machine learning, many statistical measures are introduced w.r.t two mismatched distributions, e.g., Kullback–Leibler (KL) divergence and total variation distance. The relationship between total variation distance and KL divergence has been investigated in the past years. In machine learning, \cite[p.10]{vapnik1998statistical} introduced the \textit{Bretagnolle-Huber} bound for density estimation. Vajda et al. introduced a refinement of \textit{Pinsker's} inequality in \cite{vajda1970note}. Based on \cite{vajda1970note}, Fedotov et al. derived the parametrization of the tight bound between KL-divergence and total variation distance in \cite{fedotov2003refinements}. 

While considerable attention has been devoted to investigating the total variation distance in existing literature, relatively limited studies have been placed on the classification error mismatch. The relationship between total variation distance and the classification error mismatch was derived in \cite{ney2003relationship}. In that work, Ney derived several statistical bounds on the error mismatch from the bounds for total variation distance. Nussbaum et al. provided a tight bound between $f$-divergence and error mismatch in \cite{nussbaum2013relative}, and Schlüter et al. derived the complete proof of the bound in \cite{schluter2013novel}.

While the bound derived in \cite{nussbaum2013relative,schluter2013novel} is tight when the true distribution is arbitrary, it can be refined when more information about the true distribution is obtained.
In practice, many classification tasks typically have a low Bayes error. For instance, in speech recognition, the word error rate of human speech recognition is often below 1\% \cite{wesker2005oldenburg} for a wide range of conditions \cite{lippmann1997speech}, indicating the Bayes error to be even lower. Motivated by this, in this work, we derive that the KL-divergence-based classification error bound can be refined when the Bayes error is lower than a threshold.

This paper is organized as follows: initially, we provide an overview of the fundamental concepts related to the classification error problem. Subsequently, we reexamine the proof of the general $f$-divergence-based tight bound as proposed in \cite{schluter2013novel}, and derive the local and global bounds between KL-divergence and classification error mismatch with an alternative approach, without employing the permutation method used in the original proof. Based on the local bound, we then derive a refined KL-divergence-based bound under the condition that the Bayes error remains below a certain threshold.

\section{Classification Error Mismatch}
Consider a statistical classification problem, where $pr(c,x)$ is defined as the joint true distribution for a class $c\in \mathcal{C}$ and an observation $x \in \mathcal{X}$. To simplify the discussion, we assume that $x$ is a discrete variable, where $|\mathcal{X}| > 2$ and $|\mathcal{C}|>2$. The Bayse decision rule for the classification task is defined as:
\begin{equation}
    c_*^x: =  \argmax_c pr(c,x) = \argmax_c pr(c|x).
\end{equation}
where $pr$ is the true probability. In practical applications, the true distribution is unknown. Therefore, a model distribution $q(c,x)$ is employed to estimate the true distribution. The model-based decision rule is defined as:
\begin{equation}
    c_q^x := \argmax_c q(c,x) = \argmax_c q(c|x).
\end{equation}


For a joint event $(x,c)$, given the Bayes decision rule $x \rightarrow c_*^x$ and model-based decision rule $x \rightarrow c_q^x $, the classification error counts for Bayes and model-based decision rules are defined as:
\begin{align}
    e_*(x,c) := 1 - \delta(c_*^x,c ),\quad e_q(x,c) := 1 - \delta(c_q^x,c ),
\end{align}
where $\delta(\cdot, \cdot)$ is the Kronecker delta. The local Bayes classification error $E_*\{e|x\}$ is defined as the expectation of $e(x,c)$ under the true distribution $pr(c|x)$:
\begin{align}
    \mathbb{E}\{e_*|x\} &:= \sum_c pr(c|x) e_*(x,c)= 1- pr(c_*^x|x)
\end{align}
The local model-based classification error, i.e. the expectation of error using the model-based decision rule, is defined similarly by replacing the Bayes decision rule with the model-based decision rule:
\begin{align}
    \mathbb{E}\{e_q|x\} &:= \sum_c pr(c|x)e_q(x,c) = 1- pr(c_q^x|x)
\end{align}
Note that the model distribution is only used to estimate the true distribution in the decision rule. Therefore, the expectation is still computed under the true distribution.
The effect of local classification mismatch can be represented by the local error mismatch:
\begin{align}
    \Delta_q(x) &:= \mathbb{E}\{e_q|x\} - \mathbb{E}\{e_*|x\} = pr(c_*^x|x) - pr(c_q^x|x)
\end{align}
with  $\Delta_q(x) \in [0,1]$ according to the definition. The global Bayes and model-based classification errors, $E_*$ and $E_q$, are defined as the expectation of the local errors:
\begin{align}
    E_* = \sum_x pr(x) \mathbb{E}\{e_*|x\}, \quad E_q = \sum_x pr(x) \mathbb{E}\{e_q|x\}.
\end{align}

The global error mismatch $\Delta_q$ is defined as the difference between $E_q$ and $E_*$:
\begin{equation}
    \Delta_q := E_q - E_* = \sum_x pr(x) \Delta_q(x).
\end{equation}

It is shown in \cite{ney2003relationship} that the total variation distance $V$, defined as:
\begin{equation}
    V := \frac{1}{2}\sum_{x,c} |pr(c,x) - q(c,x)|,
\end{equation}
is an upper bound of the global error mismatch, namely:
\begin{equation}
    \Delta_q \leq \sum_{x,c} |pr(c,x) - q(c,x)| = 2V
\end{equation}
This indicates that all the upper bounds for total variation distance can also be applied to $\Delta_q$, though not tight anymore. In \cite{schluter2013novel}, the bounds derived from $V$ were compared with the tight bound derived for $\Delta_q$.


\section{Relationship between Local Error Mismatch and KL-Divergence }
\subsection{Local Bound for $f$-Divergence}
For discrete variables, the \textit{$f$-divergence} from $q(c|x)$ to $pr(c|x)$ is defined as:
\begin{align}
    D_f\big (pr(c|x) \parallel q(c|x)\big) := \sum_c q(c|x) f\big (\frac{pr(c|x)}{q(c|x)}\big )
    \label{eq:localfdef}
\end{align}
where $f$ is a convex function $f: R^+ \rightarrow R$ with the specific property $f(1)=0$. The understanding of the edge cases is that:
\begin{equation}
\begin{aligned}
    f(0) = f(0^+), \quad 0f(\frac{0}{0}) = 0.
\end{aligned}
\end{equation}
Equation \eqref{eq:localfdef} is referred to as a local measure because the measure is computed for a single observation.
In \cite{schluter2013novel}, Schlüter et al. derived the global bound of $\Delta_q$ based on $f$-divergence between joint distributions $pr(c,x)$ and $q(c,x)$, which can also be applied to the proof of the local $f$-divergence. In this work, we revisit the result from \cite{schluter2013novel} and reformulate it to a local bound, with a different proof.
\begin{theorem}
The local \textit{$f$-divergence} between $pr(c|x)$ and $q(c|x)$ is tightly lower-bounded by a function of the local error mismatch $\Delta_q(x)$ in the following way:
\begin{align}
    D_f\big(pr(c|x) \parallel q(c|x) \big) \geq \frac{1}{2} \big (f(1+\Delta_q(x)) + f(1-\Delta_q(x)) \big )
    \label{eq:localfbound}
\end{align}
\end{theorem}
\begin{proof}
When $\Delta_q(x) = 0$, the right-hand side of \eqref{eq:localfbound} equals $0$, the inequality holds because of the non-negativity of $f$-divergence. In the following case, we consider the non-trivial case $\Delta_q(x) \in (0,1]$.
\begin{lemma}
    Aggregation of two summands of an f-Divergence: with $p_1, p_2, q_1, q_2 \in \mathbf{R}^+$, the following inequality holds:
    \begin{equation}
        q_1 f(\frac{p_1}{q_1}) + q_2 f(\frac{p_2}{q_2}) \geq (q_1+q_2) f(\frac{p_1+p_2}{q_1+q_2})
        \label{eq:lumping}
    \end{equation}
\end{lemma}
With this lemma, we can derive that:
\vspace{-2mm}
\allowdisplaybreaks
    \begin{align}
        &D_f\big (pr(c|x) \parallel q(c|x)\big) = \sum_c q(c|x) f\big (\frac{pr(c|x)}{q(c|x)}\big) \notag \\
                & = \sum_{c \in \{c_*^x, c_q^x\} } q(c|x) f\big (\frac{pr(c|x)}{q(c|x)}\big) + \sum_{c \in \mathcal{C}\backslash \{c_*^x, c_q^x\}}q(c|x) f\big (\frac{pr(c|x)}{q(c|x)}\big) \notag \\
                &\geq \sum_{c \in \{c_*^x, c_q^x\} } q(c|x) f\big (\frac{pr(c|x)}{q(c|x)}\big) \notag \\ 
                & \quad + \sum_{c \in \mathcal{C}\backslash \{c_*^x, c_q^x\}}q(c|x)f \big(\frac{\sum_{c \in \mathcal{C}\backslash \{c_*^x, c_q^x\}}pr(c|x)}{\sum_{c \in \mathcal{C}\backslash \{c_*^x, c_q^x\}}q(c|x)}\big) \notag \\
                & = \sum_{c \in \{c_*^x, c_q^x\} } q(c|x) f\big (\frac{pr(c|x)}{q(c|x)}\big) \notag \\
                &\quad + \big(1-q(c_*^x|x) -q(c_q^x|x)\big) f\big(\frac{1-pr(c_*^x|x) -pr(c_q^x|x)}{1-q(c_*^x|x) -q(c_q^x|x)}\big) \notag \\ 
    \end{align}
    \vspace{2mm}
    \begin{align}
    &\geq \frac{1-q(c_*^x|x)+q(c_q^x|x)}{2}f\Big(\frac{\frac{1-pr(c_*^x|x) +pr(c_q^x|x)}{2}}{\frac{1-q(c_*^x|x) +q(c_q^x|x)}{2}}\Big) \notag \\
         &\quad + \frac{1+q(c_*^x|x)-q(c_q^x|x)}{2}f\Big(\frac{\frac{1+pr(c_*^x|x) -pr(c_q^x|x)}{2}}{\frac{1+q(c_*^x|x)- q(c_q^x|x)}{2}}\Big) \notag \\
                & = \frac{1+\epsilon}{2}f\big(\frac{1-\Delta_q(x)}{1+\epsilon}\big) +  \frac{1-\epsilon}{2}f\big(\frac{1+\Delta_q(x)}{1-\epsilon}\big) \label{eq:localproofstep1} 
    \end{align}
    where $ \epsilon := q(c_q^x|x) - q(c_*^x|x) \in (0,1)$.
    Note that $f$ is convex and $f(1)=0$, based on the property of a convex function that $\frac{f(u) - f(1)}{u-1}$ is monotonically non-decreasing in $u$, let $u=\frac{1-\Delta_q(x)}{1+\epsilon}$, and $u_0 = 1-\Delta_q(x)$, we have $1>u_0 \geq u$, and therefore the following inequality holds:\\
    \scalebox{0.95}{\parbox{1.05\linewidth}{
    \begin{align}
        &\frac{f(u) - f(1)}{u-1} \leq \frac{f(u_0) - f(1)}{u_0-1} \notag \\
        &\Rightarrow (1+\epsilon)\frac{f\big(\frac{1-\Delta_q(x)}{1+\epsilon}\big)}{-\Delta_q(x) - \epsilon} \leq \frac{f\big(1-\Delta_q(x)\big)}{-\Delta_q(x)} \notag \\
        &\Rightarrow (1+\epsilon) \frac{f\big(\frac{1-\Delta_q(x)}{1+\epsilon}\big)}{\Delta_q(x) + \epsilon} \geq \frac{f\big(1-\Delta_q(x)\big)}{\Delta_q(x)} \notag \\
        &\Rightarrow (1+\epsilon) f\big(\frac{1-\Delta_q(x)}{1+\epsilon}\big) \geq f\big(1-\Delta_q(x)\big) \big(1 + \frac{\epsilon}{\Delta_q(x)}\big) \label{eq:local1term}
    \end{align}
    }}
    Similarly, let $u=\frac{1+\Delta_q(x)}{1-\epsilon}$ and $u_1 = 1+\Delta_q(x)$, we have $u\geq u_1>1$, and therefore we obtain the following inequality:\\
    \scalebox{0.95}{\parbox{1.05\linewidth}{
        \begin{align}
        &\frac{f(u) - f(1)}{u-1} \geq \frac{f(u_1) - f(1)}{u_1-1} \notag \\
        &\Rightarrow (1-\epsilon) f\big(\frac{1+\Delta_q(x)}{1-\epsilon}\big) \geq f\big(1+\Delta_q(x)\big) \big(1 + \frac{\epsilon}{\Delta_q(x)}\big) \label{eq:local2term}
    \end{align}
    }}
    Then, by substituting \eqref{eq:local1term} and \eqref{eq:local2term} into \eqref{eq:localproofstep1}, we have:\\
    \scalebox{0.95}{\parbox{1.05\linewidth}{
    \allowdisplaybreaks
    \begin{align}{2}
        &D_f\big (pr(c|x) \parallel q(c|x)\big) \notag \\
        &\geq \frac{1}{2} \Big( \underbrace{f\big(1+\Delta_q(x)\big) + f\big(1-\Delta_q(x)\big)}_{\geq 2f\big(\frac{1+\Delta_q(x)+ 1-\Delta_q(x)}{2}\big) =2f(1) =0}\Big) \big(1 + \frac{\epsilon}{\Delta_q(x)}\big) \notag \\
       & \geq \lim_{\epsilon \rightarrow 0^+} \frac{1}{2} \Big(f\big(1+\Delta_q(x)\big) + f\big(1-\Delta_q(x)\big)\Big) \big(1 + \frac{\epsilon}{\Delta_q(x)}\big) \notag \\
       & = \frac{1}{2} \Big(f\big(1+\Delta_q(x)\big) + f\big(1-\Delta_q(x)\big)\Big).
        \label{eq:step2}
    \end{align}
    }}
    Therefore, \eqref{eq:localfbound} is proved.
    When $\Delta_q(x)=0$, the equality can be obtained by $pr(c|x) = q(c|x), \forall c$. When $\Delta_q \in (0,1]$, the equality of the bound can be obtained by the following parametrized distribution with $\lambda \in (0.5,1]$:\\
    \allowdisplaybreaks
    \scalebox{0.95}{\parbox{1.05\linewidth}{
\begin{align}
\begin{split}
    pr(c|x) &= \left \{ \begin{array}{ll}
      \lambda,   & c=c_1 \\
      1-\lambda,  & c=c_2\\
      0, & \text{otherwise}
          \end{array} \right . ,\\
        q(c|x) &= \lim_{\epsilon \rightarrow 0^+}\left \{\begin{array}{ll}
         0.5 - \epsilon, & c= c_1\\
         0.5 + \epsilon, & c=c_2\\
         0, &\text{otherwise} 
         \end{array}     \right . ,
\label{eq:tightdist}
\end{split}
\end{align}
}}
where $c_1$ and $c_2$ are two different classes, and $\epsilon$ ensures that $c_q^x=c_2 \neq c_1 = c_*^x$. With the distributions discussed above, the tightness of the bound is verified. 
\end{proof}

\subsection{Local Bound for KL-Divergence}
The KL-divergence is obtained by setting $f(u) = u\log u$. The associated lower bound becomes:
\begin{align}
    &D_\text{KL}(pr(c|x) \parallel q(c|x)) \geq B\big(\Delta_q(x)\big)
    \label{eq:localkl}
\end{align}
where $B$ is defined as:
\begin{equation}
    B(u) = \frac{1}{2}\big((1+u)\log(1+u) + (1-u)\log (1-u) \big)
\end{equation}
Note that $B$ is also a convex function and $B(0) = 0$.

\section{Relationship between Global Error Mismatch and KL-Divergence}
In machine learning tasks, the performance of the model is usually not measured on one single observation or data point, but on the whole dataset. Therefore, in this section, we investigate the bounds for the global error mismatch $\Delta_q$. 
\subsection{Conditional KL-Divergence}
Since $B\big(\Delta_q(x)\big)$ is the lower bound of the local KL-divergence and $B$ is convex, the expectation of local KL-divergence under $pr(x)$, namely, the conditional KL-divergence is lower-bounded by $B\big(\Delta_q\big)$:\\
\scalebox{0.95}{\parbox{1.05\linewidth}{
\begin{align}
    \sum_{x} &pr(x) D_\text{KL}\big(pr(c|x) \parallel q(c|x)\big) \geq \sum_x pr(x) B\big(\Delta_q(x)\big) \notag \\
    & \geq B\big(\sum_x pr(x)\Delta_q(x)\big) = B(\Delta_q) 
    \label{eq:expectbound}
\end{align}
}}
The equality for $\Delta_q \in (0,1]$ can be obtained with such a distribution that for each $x$, the conditional probabilities are like in \eqref{eq:tightdist} with $\lambda \in (0.5,1]$. When $\Delta_q=0$, the equality can be obtained by $pr(c|x)=q(c|x), \forall x,c$. 


\subsection{Joint KL-Divergence}
Now we consider another global measure, KL-divergence $D_\text{KL}(pr\parallel q)$ between joint distributions $pr(c,x)$ and $q(c,x)$:
\begin{equation}
    D_\text{KL}(pr\parallel q) := \sum_{x,c} pr(c,x) \log \frac{pr(c,x)}{q(c,x)}.
\end{equation}
In \cite{schluter2013novel}, the bound between $D_\text{KL}(pr\parallel q)$ and $\Delta_q$ was proved by utilizing the permutation operation. Here, we provide an alternative proof. To this end, we show that $D_\text{KL}(pr \parallel q)$ is lower-bounded by the conditional KL-divergence:
\allowdisplaybreaks
\begin{align}
    &D_\text{KL}(pr \parallel q) = \sum_{x,c} pr(c,x) \big (\log \frac{pr(c|x)}{q(c|x)} + \log \frac{pr(x)}{q(x)}) \notag \\
    & = \sum_x pr(x) \sum_c pr(c|x) \log \frac{pr(c|x)}{q(c|x)} + \sum_x pr(x) \log \frac{pr(x)}{q(x)}\notag \\
    &\geq \sum_x pr(x) \sum_c pr(c|x) \log \frac{pr(c|x)}{q(c|x)}     \label{eq:condtojoint} 
\end{align}

Combining \eqref{eq:expectbound} and \eqref{eq:condtojoint}, it is observed that $D_\text{KL}(pr\parallel q)$ is also lower-bounded by $B(\Delta_q)$:
\begin{equation}
     D_\text{KL}(pr \parallel q) \geq  B(\Delta_q). \label{eq:globalbound}
\end{equation}
The equality is obtained when having the same conditional distributions as in the case of conditional KL-divergence, with additional condition $pr(x)=q(x)$.

\section{Refined Bounds with Constraints on $E_*$}
The bound \eqref{eq:globalbound} is tight when the true distribution $pr(c,x)$ is unconstrained. However, the bound can be refined if the true distribution is subject to some constraints. One typical condition in machine learning tasks is that the Bayes error is low. For instance, the word-level classification error of human speech recognition can be below 1\% \cite{lippmann1997speech}, indicating the Bayes error to be even lower. Inspired by this, we consider a system with low Bayes error $E_*\leq t < 0.5$, where $t$ is an estimated threshold. Under this constraint, a refined bound is given as:
\begin{theorem}
\label{theorem:refinedbound}
    When $E_*\leq t< 0.5$, $D_\text{KL}(pr||q)$ is lower-bounded by the following function of $\Delta_q$,\\
\begin{equation}
    \begin{aligned}
        D_\text{KL}(pr||q) \geq 
        \left \{ \begin{array}{lr}
     (\Delta_q+ 2t) B\big (\frac{\Delta_q}{\Delta_q+2t}\big )&, \text{for }\Delta_q \in [0, 1-2t)\\
        B(\Delta_q) &,\text{for }\Delta_q \in [1-2t, 1]
    \end{array}
\right .
\label{eq:constrainedbound}
\end{aligned}
\end{equation}
\end{theorem}


To prove Theorem \ref{theorem:refinedbound}, we first introduce two lemmas, followed by their proofs.
\begin{lemma}
\label{lemma:x0}
    When $\Delta_q< 1-2t$, there is at least one observation $x_0$ with $pr(x_0) > 0$ such that $c_*^{x_0} = c_q^{x_0}$ holds.
\end{lemma}
\begin{proof}
    For the sake of contradiction, suppose for all $x$, $c_*^x \neq c_q^x$, then for each local error mismatch $\Delta_q(x)$, we have
\begin{equation}
\begin{aligned}
    \Delta_q(x) = pr(c_*^x|x) - pr(c_q^x|x)  &\geq pr(c_*^x|x) - \big(1- pr(c_*^x|x)\big) \\
    &= 1- 2\mathbb{E}\{e_*|x\}\label{eq:delta_and_accu}
    \end{aligned}
\end{equation}
Therefore, $\Delta_q$ is lower-bounded by: 
\begin{align}
    \Delta_q = \sum_x pr(x) \Delta_q(x)&\geq 2\sum_x pr(x)pr(c_*^x|x) -1 \notag \\
    &=1- 2E_* \geq 1-2t,
\end{align}
which contradicts the condition on $\Delta_q$.
\end{proof}
\vspace{-2mm}
Note that since $c_*^{x_0} = c_q^{x_0}$, $\Delta_q(x_0)=0$. Without loss of generality, we assume that for all the other observations $x \neq x_0$, $c_*^x \neq c_q^x$. When $pr(x_0) = 1$, $\Delta_q= pr(x_0) \Delta_q(x_0) =0$, and the right-hand side of \eqref{eq:constrainedbound} is $0$. Therefore \eqref{eq:constrainedbound} holds. In the following discussion, we assume $pr(x_0)<1$. $\Delta_q$ can be rewritten as:
\allowdisplaybreaks
\begin{align}
    \Delta_q &= pr(x_0) \Delta_q(x_0) + \sum_{x \neq x_0} pr(x) \Delta_q(x) \notag \\
    & = 0 + \big (1-pr(x_0)\big) \sum_{x \neq x_0} \frac{pr(x)}{1-pr(x_0)} \Delta_q(x) \notag \\
    & = \big (1-pr(x_0)\big) \Tilde{\Delta}_q(x_0)
    \label{eq:px0}
\end{align}
where $ \Tilde{\Delta}_q(x_0)$ is the expected classification error for the renormalized true distribution $\Tilde{pr}(x)$ without $x_0$.\\
\scalebox{0.95}{\parbox{1.05\linewidth}{
\begin{align}
    \Tilde{pr}(x) = \frac{pr(x)}{1-pr(x_0)}, \Tilde{\Delta}_q(x_0)= \sum_{x \neq x_0} \Tilde{pr} (x)\Delta_q(x)
    \label{eq:renormDelta}
\end{align}
}}

\begin{lemma}
    $\forall \Delta_q \in [0,1-2t)$, $\Tilde{\Delta}_q(x_0)$ is lower-bounded by:\\
   \label{lemma:mindelta}

\begin{align}
    \Tilde{\Delta}_q(x_0)\geq \Delta_q \big(1+ \frac{1-2t-\Delta_q}{2(t-\mathbb{E}\{e_*|x_0\})+\Delta_q}\big ) \geq \frac{\Delta_q}{2t+\Delta_q}.
\end{align}
\vspace{-2mm}
\end{lemma}
\begin{proof}
when $\Delta_q=0$, the Lemma obviously holds. Therefore, we consider the interval $ \Delta_q \in (0, 1-2t)$. According to the definition of $\Delta_q$, we have:\\
\scalebox{0.95}{\parbox{1.05\linewidth}{
\begin{align}
    \Delta_q 
      & = \left(1-pr(x_0)\right) \Tilde{\Delta}_q(x_0) \Rightarrow  \quad pr(x_0)= 1- \frac{\Delta_q}{\Tilde{\Delta}_q(x_0)}. \label{eq:prx0}
\end{align}
}}
Since $\Delta_q < 1-2t$, we also have:
\begin{align}
    \big(1-pr(x_0)\big) \Tilde{\Delta}_q(x_0) = \Delta_q   <1-2t.
    \label{eq:deltasmall}
\end{align}
Meanwhile, according to the constraint $E_* \leq t$, we have:
\begin{align}
    E_* = pr(x_0) \mathbb{E}\{e_*|x_0\} + \big(1-pr(x_0)\big)\Tilde{E}_*(x_0) \leq t \label{eq:exo}.
\end{align}
where $\Tilde{E}_*(x_0)$ is defined as:
\begin{equation}
    \Tilde{E}_*(x_0) = \sum_{x \neq x_0 } \Tilde{pr}(x) \mathbb{E}\{e_*|x\}.
\end{equation}
According to \eqref{eq:delta_and_accu} and the definition of $\Tilde{\Delta}_q(x_0)$ in \eqref{eq:renormDelta}, $ \Tilde{E}_*(x_0)$ and $\Tilde{\Delta}_q(x_0)$ have the following relationship:
\begin{align}
    \Tilde{\Delta}_q(x_0) \geq1- 2\Tilde{E}_*(x_0) \Leftrightarrow \Tilde{E}_*(x_0) \geq \frac{1-\Tilde{\Delta}_q(x_0)}{2} \label{eq:Tildee}
\end{align}
According to \eqref{eq:deltasmall} and \eqref{eq:Tildee}, we have:
\begin{align}
    (1-pr(x_0)) (1-2\Tilde{E}_*(x_0)) < 1-2t \notag \\
    \Rightarrow (1-pr(x_0)) \Tilde{E}_*(x_0) > t-\frac{pr(x_0)}{2} \label{eq:e0}
\end{align}
By substituting \eqref{eq:e0} into \eqref{eq:exo}, we obtain that $\mathbb{E}\{e_*|x_0\}< \frac{1}{2}$:
\begin{align}
    pr(x_0) \mathbb{E}\{e_*|x_0\} &\leq t - (1-pr(x_0))\Tilde{E}_*(x_0)< t - t+ \frac{pr(x_0)}{2} \notag\\
    \Rightarrow \mathbb{E}\{e_*|x_0\}& < \frac{1}{2}.
\end{align}
To obtain the bound for $\Delta_q$, by substituting \eqref{eq:prx0} into \eqref{eq:exo}, combined with \eqref{eq:Tildee}, the following inequality can be derived.
\begin{align}
    pr(x_0) \mathbb{E}\{e_*|x_0\} + \big(1-pr(x_0)\big) \frac{1-\Tilde{\Delta}_q(x_0)}{2} \leq t \notag \\
    \Rightarrow \big(1- \frac{\Delta_q}{\Tilde{\Delta}_q(x_0)}\big) \mathbb{E}\{e_*|x_0\} + \big(\frac{\Delta_q}{\Tilde{\Delta}_q(x_0)}\big) \frac{1-\Tilde{\Delta}_q(x_0)}{2} \leq t \notag \\
    \Rightarrow (2t-2\mathbb{E}\{e_*|x_0\}+\Delta_q) \Tilde{\Delta}_q(x_0)\geq \Delta_q (1 - 2\mathbb{E}\{e_*|x_0\}) \label{eq:deltaqfinal}
\end{align}
\vspace{1mm}
\begin{figure}[htb!]
  \centering
  \includegraphics[height=4cm]{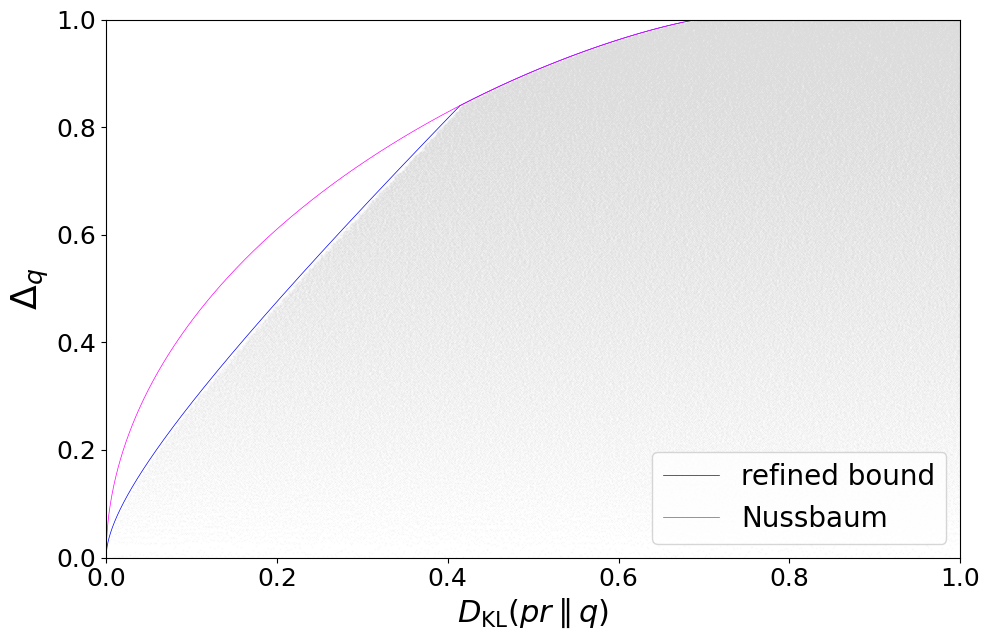}
   \caption{Comparison of the Nussbaum bound \cite{nussbaum2013relative} and the refined bound in this paper. The simulations in the upper figure are under the constraint $E_* \leq 0.08$. The grey dots refer to the simulation points.}
  \label{fig:kl_compare}
\end{figure}
Now we prove that
\begin{equation}
    2t-2\mathbb{E}\{e_*|x_0\}+\Delta_q >0.
    \label{eq:positivecond}
\end{equation}
For the sake of contradiction, if $2t-2\mathbb{E}\{e_*|x_0\}+\Delta_q \leq 0 $, the left-hand side of \eqref{eq:deltaqfinal} is less than or equal to $0$, while the right-hand side of the inequality, $\Delta_q (1 - 2\mathbb{E}\{e_*|x_0\}) > 0$ because $\mathbb{E}\{e_*|x_0\}< 0.5 $ and $\Delta_q >0$, which contradicts to the inequality. Therefore, \eqref{eq:positivecond} holds. By dividing both sides of \eqref{eq:deltaqfinal} by \eqref{eq:positivecond}, we have:
\begin{equation}
    \begin{aligned}
    \Tilde{\Delta}_q(x_0) &\geq \frac{\Delta_q (1 - 2\mathbb{E}\{e_*|x_0\})}{2t-2\mathbb{E}\{e_*|x_0\}+\Delta_q} \\
    &= \Delta_q \big(1+ \frac{1-2t-\Delta_q}{2(t-\mathbb{E}\{e_*|x_0\})+\Delta_q}\big )
    \end{aligned}
\end{equation}
Because $\Delta_q<1-2t$ and $\mathbb{E}\{e_*|x_0\} \geq 0$, the right-hand side of the inequality obtains minimum when $\mathbb{E}\{e_*|x_0\}=0$, i.e.
\begin{equation}
     \Tilde{\Delta}_q(x_0)\geq \Delta_q \big(1+ \frac{1-2t-\Delta_q}{2(t-\mathbb{E}\{e_*|x_0\})+\Delta_q}\big ) \geq \frac{\Delta_q}{2t+\Delta_q}
     \label{eq:deltalowerbound}
\end{equation}
\end{proof}
Based on Lemma \ref{lemma:x0} and Lemma \ref{lemma:mindelta}, we provide the proof of Theoream \ref{theorem:refinedbound} as follows:
\begin{proof}[Proof of Theorem \ref{theorem:refinedbound}]
\allowdisplaybreaks
When $\Delta_q=0$, the inequality holds because of the non-negativity. When $\Delta_q \in (0,1-2t)$, we have:
\begin{align*}
    &D_\text{KL}(pr \parallel q) \geq \sum_x pr(x) \sum_c pr(c|x) \log \frac{pr(c|x)}{q(c|x)} \Big(\text{c.f. } \eqref{eq:condtojoint}\Big )\\
     &= pr(x_0) D_\text{KL}\big( pr(c|x_0) \parallel q(c|x_0) \big)  \\
     & \qquad \quad + \sum_{x \neq x_0 } pr(x) D_\text{KL}\big( pr(c|x) \parallel q(c|x) \big) \\
    &= pr(x_0) D_\text{KL}\big( pr(c|x_0) \parallel q(c|x_0) \big)\\
    & \quad + \big (1-pr(x_0)\big ) \sum_{x \neq x_0}\Tilde{pr}(x) D_\text{KL}\big( pr(c|x) \parallel q(c|x) \big) \\
    & \geq \underbrace{pr(x_0) B\big(\Delta_q(x_0)\big)}_{= 0\text{, c.f. Lemma \ref{lemma:x0}}} + \big(1-pr(x_0)\big) \underbrace{\sum_{x \neq x_0}\Tilde{pr}(x) B\big(\Delta_q(x)\big)}_{\text{apply Jensen inequality}}\\
    &\geq \big (1-pr(x_0)\big ) B \big(\underbrace{\sum_{x \neq x_0} \Tilde{pr}(x) \Delta_q(x)}_{= \Tilde{\Delta}_q(x_0), \text{ c.f. \eqref{eq:renormDelta}}} \big ) \\
\end{align*}

\begin{align*}
\vspace{0.5mm}
 &= \big(1-pr(x_0)\big) B\big(\Tilde{\Delta}_q(x_0)\big) = \Delta_q \frac{B\big(\Tilde{\Delta}_q(x_0) \big)}{\Tilde{\Delta}_q(x_0)}  \Big (\text{c.f. \eqref{eq:px0}} \Big)\\
    & \geq \Delta_q \frac{B(\frac{\Delta_q}{2t+\Delta_q})}{\frac{\Delta_q}{2t+\Delta_q}} = (\Delta_q + 2t) g\big(\frac{\Delta_q}{\Delta_q+2t}\big)  \tag{\stepcounter{equation}\theequation}\\
    & \big(\text{c.f. \eqref{eq:deltalowerbound}, and } B\text{ is convex, } \frac{B(u)-g(0)}{u-0} \text{ is monotonically}\\
    &\text{non-decreasing; equality for } \Tilde{\Delta}_q(x_0)=\frac{\Delta_q}{\Delta_q+2t} \big ) 
\end{align*}
For the second segment $\Delta_q \in [1-2t,1]$, let $pr(c|x)$ and $q(c|x)$ be distributions given in \eqref{eq:tightdist} for each $x$. In this case, 
\begin{align}
    \Delta_q &= 2\lambda -1 \Rightarrow \lambda = \frac{1+\Delta_q}{2} \in [1-t,1],\\
    E_* &= 1- \lambda \leq t
\end{align}
which shows that these distributions are valid under the constraint. Therefore, \eqref{eq:globalbound} is still the tightest bound.
\end{proof}
\vspace{-1.5mm}
For the first segment where $\Delta_q \in [0,1-2t)$, equality is achieved through a particular selection of distributions. Effectively, there are two observations $x_1,x_2$, and $pr(x) = 0$ for $x \notin \{x_1,x_2\}$. Given a parameter $\lambda \in [0.5, 1-t)$, the true and model distributions for observation $x_1$ and $x_2$ are parametrized as follows:\\
\scalebox{0.95}{\parbox{1.05\linewidth}{
\begin{equation}
\begin{aligned}
    pr(x_1) &=1- \frac{t}{1-\lambda},  pr(c|x_1) = q(c|x_1) = \left \{\begin{array}{ll}
         1, &  c= c_1\\
         0, &\text{otherwise} 
    \end{array}
    \right .\\
   pr(x_2) &= \frac{t}{1-\lambda}, \quad pr(c|x_2) =  \left \{\begin{array}{ll}
         \lambda, &  c= c_1\\
         1- \lambda, &  c=c_2\\
         0, &\text{otherwise} 
         \end{array}     \right .\\
    q(c|x_2) &= \lim_{\epsilon \rightarrow 0^+}\left \{\begin{array}{ll}
         0.5 - \epsilon, & c= c_1\\
         0.5 + \epsilon, & c=c_2\\
         0, &\text{otherwise} 
         \end{array}     \right .
\end{aligned}
\end{equation}
}}
Figure \ref{fig:kl_compare} shows the comparison of the Nussbaum bound \eqref{eq:globalbound} and the derived bound \eqref{eq:constrainedbound} with simulation results, on the constraint $E_* \leq 0.08$. The simulation was conducted by generating various distribution pairs $(pr,q)$ until all the reachable areas were covered. Each grey dot represents the result of a single simulation.
The simulation results verify that \eqref{eq:constrainedbound} is tight under the constraint $E_*\leq t$ for $\Delta_q \in [0,1]$, providing an improved bound compared to Nussbaum bound.

\section{Conclusion}
In this work, we investigated the relationship between the Kullback-Leibler divergence and the classification error mismatch, which is introduced by replacing the unknown true distribution with the distribution from the model in Bayes decision rule. We started by revisiting the statistical bound with unconstrained true distributions from previous works and offered an alternative proof. Motivated by the assumption that the Bayes error is typically low in machine learning tasks, we further derived the refined bound on the error mismatch under the constraint that Bayes error is lower than a threshold. The analytical results were supported by simulations.

\bibliographystyle{IEEEtran}
\bibliography{mybib}

\end{document}